\documentclass[12pt,a4paper]{article}

\usepackage{amsfonts}
\usepackage{amssymb}
\usepackage{amsmath}
\usepackage{amsthm}
\usepackage{amscd} 
\usepackage{graphicx}
\usepackage{color}
\usepackage[colorlinks=true,citecolor=black,linkcolor=black,urlcolor=blue]{hyperref}
\usepackage{comment}

\theoremstyle{plain}
\newtheorem{theorem}{Theorem}[section]
\newtheorem{lemma}{Lemma}[section]
\newtheorem{corollary}{Corollary}[section]
\newtheorem{proposition}{Proposition}[section]

\theoremstyle{definition}
\newtheorem{definition}{Definition}[section]

\theoremstyle{remark}

\newcommand{\beqa}{\begin{eqnarray}}
\newcommand{\eeqa}{\end{eqnarray}}

\newcommand{\dd}{\mathrm{d}}
\newcommand{\prt}{\partial}
\newcommand{\be}{\begin{equation}}
\newcommand{\ee}{\end{equation}}
\newcommand{\bea}{\begin{eqnarray}}
\newcommand{\eea}{\end{eqnarray}}
\def\ba#1\ea{\begin{align}#1\end{align}}
\def\bas#1\eas{\begin{align*}#1\end{align*}}
\newcommand{\bd}{\begin{displaymath}}
\newcommand{\ed}{\end{displaymath}}

\newcommand{\mt}{\textrm}
\newcommand{\mb}{\mathbb}
\newcommand{\mc}{\mathcal}

\definecolor{r}{rgb}{1,0,0}

\title{Next-to-leading order in the large $N$ expansion\\of the multi-orientable 
random tensor model}

\author{Matti Raasakka$^{a}$ and Adrian Tanasa$^{a,b}$\\
{\tt\small matti.raasakka@lipn.univ-paris13.fr, 
adrian.tanasa@ens-lyon.org}}

\date{}

\begin{document}

\maketitle

\begin{abstract}
In this paper we analyze in detail the next-to-leading order (NLO) of the recently obtained 
large $N$ expansion for the multi-orientable (MO) tensor model.
From a combinatorial point of view, we find the class of Feynman tensor graphs contributing 
to this order in the expansion.
Each such NLO graph is characterized by the property that it contains a certain non-orientable ribbon subgraph (a non-orientable {\it jacket}). 
Furthermore, we find the radius of convergence and 
the susceptibility exponent of the NLO series for this model. These results represent a first 
step towards the larger goal of defining an appropriate double-scaling limit for 
the MO\ tensor model. 
\end{abstract}

\newpage

\section{Introduction and motivation}
Following the success of matrix models in providing a partition function for 2-dimensional random discrete geometries and simplicial quantum gravity \cite{DGZ}, tensor models were introduced in the beginning of the 90's to generalize this approach to higher dimensions \cite{sasakura,ambjorn}. On the other hand, attempts to formulate a covariant path integral representation of loop quantum gravity led through spin foams to group field theory \cite{RR}, a generalization of tensor models obtained by replacing the tensor indices by representation labels of a Lie group, which encode additional quantum geometrical data. However, only recently a proper control over the perturbative series of these models has been gained by introducing restrictions on the class of simplicial geometries that arise from the model as its Feynman graphs. In particular, for the case of colored tensor models, and originally the colored group field theory introduced in \cite{gurau}, the so-called color labels of the graph edges impose special restrictions on their Feynman graphs, so that much of the topological data is encoded into the labels. This allowed to express the Feynman integrals of colored tensor models in terms of quantities associated to the combinatorics of the associated simplicial geometries \cite{Razvan}, reminiscent of the famous topological expansion of matrix models. (See \cite{GurauRyan} for a recent review on colored tensor models. A good general overview of the relations between loop quantum gravity, spin foams and group field theory can be found in \cite{Rovelli}.)

Recently, the large $N$ expansion and the double-scaling limit of colored tensor models have been under intensive study. It has been shown that the $1/N$ expansion is dominated in any number of dimensions by a subclass of triangulations with trivial topology, the so-called {\it melonic graphs} \cite{Razvan, GR}. The critical behavior of the leading order series was derived in \cite{BGRR}, which enabled a further study of the next-to-leading order (NLO) series in \cite{KOR}. Very recently, two different approaches, via a combinatorial classification of graphs \cite{GS} and via a resummation of the perturbative series \cite{DGR}, have made it possible to rigorously define a double-scaling limit for the colored tensor models, and study the family of leading order graphs in the limit.

These successes of the colored tensor model prompt us to ask if the results mentioned above can be extended to other models with less restrictions on the simplicial geometries than in the colored models. In this direction, the multi-orientable (MO) tensor model, introduced originally in \cite{mo} (see also \cite{praa} for a short review), represents an interesting generalization with respect to the colored models. In general, the partition function of a tensor model is defined so that the Feynman graphs of the model (which correspond to the individual summands in the pertubative series of the partition function) are dual to members in some class of closed simplicial complexes. The exact properties of this class depend on the restriction that the specific tensor model imposes on the gluing of simplices. In particular, the Feynman graphs of the MO model span a strictly larger class of 3d complexes than that of the corresponding colored model \cite{mo}, as illustrated in Fig.\ \ref{fig:class}, and thus enable us to test the universality of the properties of colored models uncovered in the recent literature.
\begin{figure}
\centering
\includegraphics[width=0.8\columnwidth]{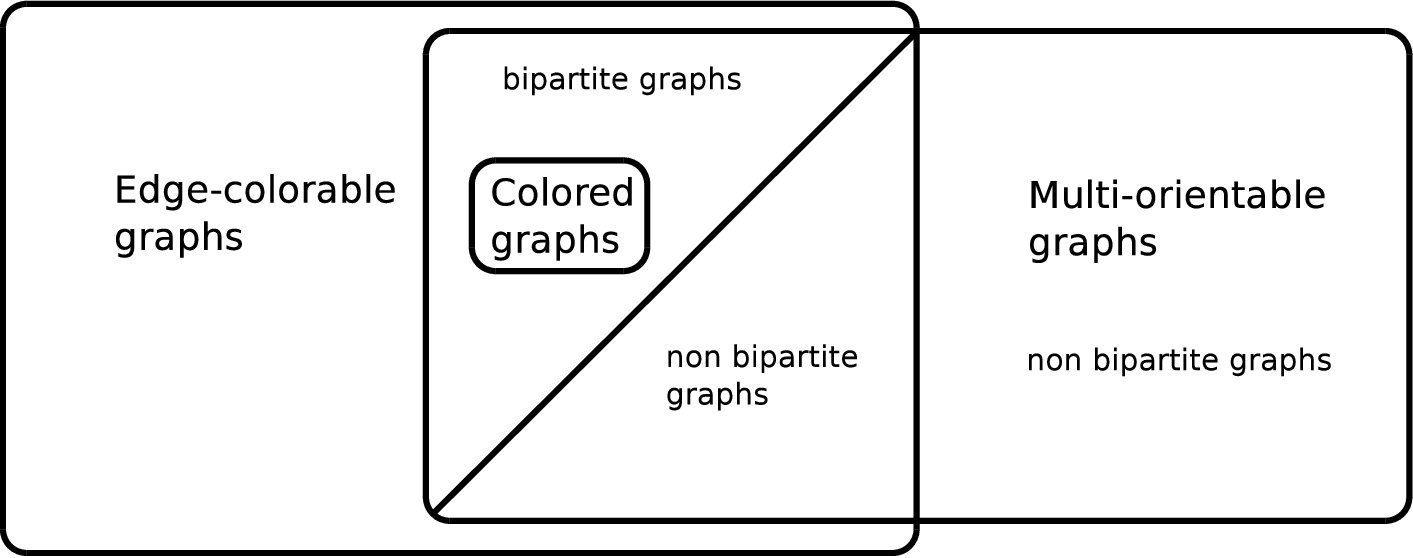}
\caption{\label{fig:class}Relations between different classes of tensor graphs \cite{Stephane}.}
\end{figure}

Motivated by these considerations, we study in the present paper the NLO graphs of the MO random tensor model. 
We find that these NLO graphs are distinct from the NLO graphs of the colored model, which is due to the wider class of graphs allowed by the MO model. Specifically, the NLO sector of the MO model is made by tensor graphs, which contain a non-orientable jacket, whereas all graphs of the colored model contain only orientable jackets. After a short review on MO tensor model in the following section, we classify the NLO graphs, and study the critical behavior of the NLO series in sections $3$ and $4$ below. In particular, we find the values of the radius of convergence and the susceptibility exponent.
The last section is dedicated to some concluding remarks and some perspectives for future research.

\section{The multi-orientable tensor model}

\subsection{Definition of the model; large $N$ expansion}

We recall here the definition of the multi-orientable (MO) tensor model and the implementation of the large $N$ expansion of this model. This follows \cite{Stephane}.

The random variables described by the MO tensor model are a complex-valued rank-3 tensor $\phi_{ijk}$ and its complex conjugate $\overline{\phi}_{ijk}$, where the indices take values in a range of integers $i,j,k=1,\ldots,N$.\footnote{This generalizes naturally to different ranges for the different indices \cite{Stephane}, but in this paper we will only consider tensors, whose indices take values in the same range.} The measure of the model consists of two parts: the Gaussian `kinetic' part, which yields an independently identically distributed Gaussian covariance for the tensor components, and a non-Gaussian perturbation controlled by a parameter $\lambda\in\mb{R}$. The partition function can be written as
\ba
	Z(\lambda,N) = \int \dd\phi \dd\overline{\phi}\ e^{-(S_0(\phi,\overline{\phi};N) + \lambda S_p(\phi,\overline{\phi};N))} \,,
\ea
where
\ba
	S_0(\phi,\overline{\phi};N) = \sum_{i,j,k\in [1,N]} \overline{\phi}_{ijk} \phi_{ijk} \quad\mt{and}\quad
	S_p(\phi,\overline{\phi};N) = \sum_{\substack{i,j,k,\\i',j',k'\\\in [1,N]}} \phi_{ijk} \overline{\phi}_{kj'i'} \phi_{k'ji'} \overline{\phi}_{k'j'i} 
\,.
\ea
Here, the partition function measure $\dd\phi \dd\overline{\phi}$ is the usual product of Lebesgue measures over the individual tensor components. As in other tensor models modelling random geometry, the pairing of indices in the perturbation term reflects the combinatorial identification of the boundary edges of the boundary triangles of a tetrahedron.

From the form of the partition function it follows that the Feynman graphs of the MO\ model are built from stranded 4-valent vertices (corresponding to the pertubation term in the action) and lines between the vertices (corresponding to the 
quadratic part of the action) illustrated in Fig.\ \ref{fig:propvert}, where each strand corresponds to one of the three indices of the tensor. Due to the choice of the action, the connectivity of the strands again reflects the combinatorial structure of a tetrahedron. The orientation of a line is marked by $+$ and $-$, corresponding to the tensor variable and its complex conjugate, respectively.
\begin{figure}
\centering
\def\svgwidth{0.5\columnwidth}
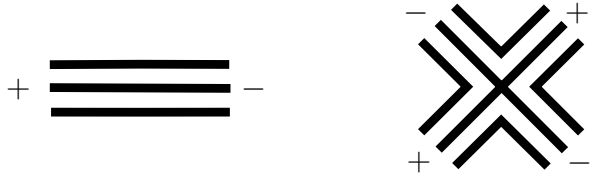
\caption{\label{fig:propvert}The stranded line and the stranded vertex of the Feynman graphs of MO\ tensor model.}
\end{figure}
The Feynman graphs of the model are obtained by gluing the vertices together, so that they are connected by the lines, i.e., so that each line has one end labelled by $+$ and the other by $-$. In the following, we will often drop these $\pm$ labels from the graphs for convenience, but they can of course be put back in whenever needed.

The strands of any graph built according to the above rules can be classified into different types according to their position in the vertices. First, we will call the strands that run the ``middle'' of the vertices the \emph{inner strands}. It is clear from the gluing rules that inner strands always connect to inner strands. The other strands are called \emph{outer}, and they can be further divided into two types, so that in any vertex the strands on the opposite sides of the vertex are of the same type. It was shown in \cite{Stephane} that also these types only connect to themselves. Accordingly, any closed loop formed by a strand, which we will call a \emph{face} in the following, has one of the three types (see Fig.\ \ref{fig:strandtypes} for an example, where the different types for the MO\ vertex and an example graph are denoted by different colors; the inner strands are green, while the outer strands are blue and red).
\begin{figure}
\centering
\def\svgwidth{0.65\columnwidth}
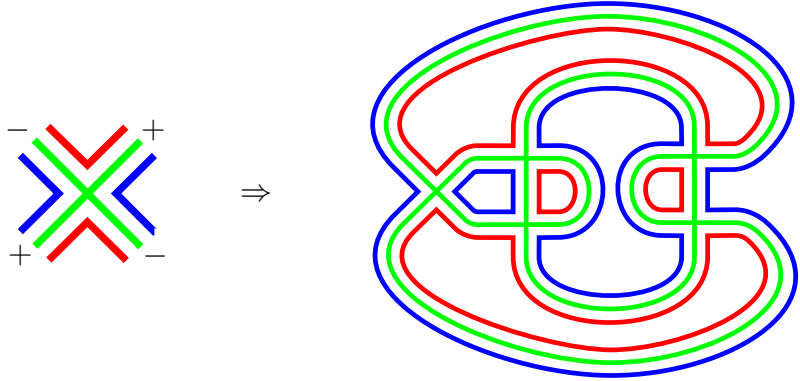
\caption{\label{fig:strandtypes}Three different types of strands in the MO\ vertex and an example graph distinguished by colors.}
\end{figure}

Taking advantage of the strand types, one may define the important notion of \emph{jackets} of a graph. These are the ribbon graphs formed by any two of the three types of strands. In Fig.\ \ref{fig:jacketex} we illustrate the three jackets of the example graph from the previous Fig.\ \ref{fig:strandtypes}.
\begin{figure}
\centering
\def\svgwidth{0.9\columnwidth}
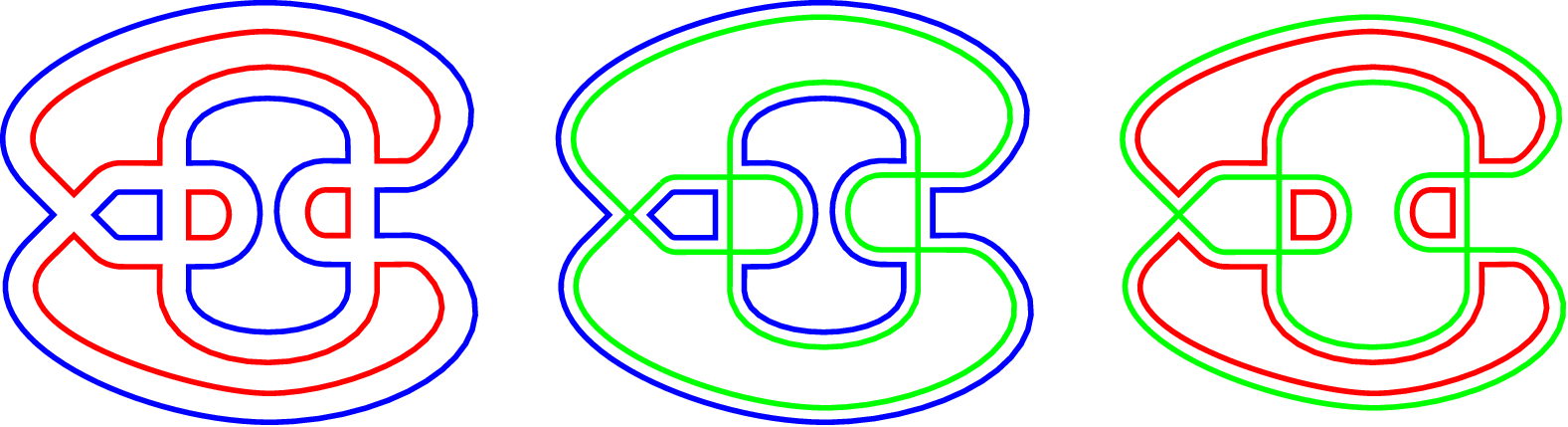
\caption{\label{fig:jacketex}The three jackets of the example MO graph in Fig.\ \ref{fig:strandtypes}.}
\end{figure}
The jackets of the MO\ tensor model are dual to triangulations of 2-dimensional surfaces, but unlike for colored models, these may be non-orientable. Therefore, the genera of these surfaces, defined in the following through the Euler characteristic formula, take also half-integer values.

We are mainly interested in the free energy $E(\lambda,N) = -\ln Z(\lambda,N)$ 
of the MO model.
This is obtained by summing over the connected vacuum Feynman graphs. Let us denote the number of vertices of an MO\ graph by $V$, the number of lines by $L$, and the number of faces by $F$. For a vacuum graph of the MO\ model we have $L=2V$, due to 4-valent vertices. Denote the jackets by $J$. Then, as show in \cite{Stephane}, the amplitude of a vacuum MO\ graph may be written as
\ba
	\mc{A} = \lambda^{V} N^{3-\omega} \,.
\ea
The \emph{degree} $\omega$ (which controls the large $N$ expansion) reads
\ba\label{eq:omega}
	\omega = \sum_{J} g_J = 3 + \frac{3}{2}V - F
\ea
using the Euler characteristic formula for the jacket genera $g_J = 1 - \frac{1}{2}(F_J - L_J + V_J)$, where $L_J=L$ and $V_J=V$ for all jackets $J$, and the fact that each face belongs to exactly two jackets.

\subsection{Leading order series: radius of convergence \& susceptibility exponent}

It was shown in \cite{Stephane} that the leading order connected vacuum graphs, for which $\omega=0$, are the same \emph{melonic graphs} as in colored models \cite{Razvan}. These are the graphs obtained from the so-called ``elementary melon'' graph by insertions of melonic 2-point subgraphs --- the so-called 3-dipole moves, see Figs.\ \ref{fig:3dip} and \ref{fig:melons}. Thus, the analysis of the leading order (LO) series of the MO\ model follows the one of the colored model (see \cite{BGRR} for the details). 
\begin{figure}
\centering
\def\svgwidth{0.5\columnwidth}
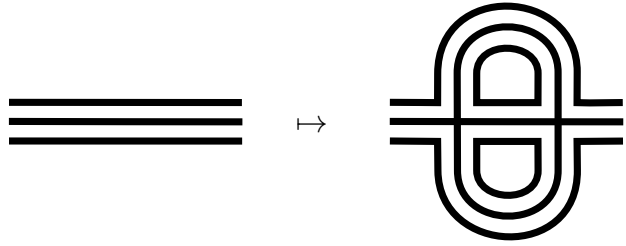
\caption{\label{fig:3dip}An insertion of the elementary melonic 2-point subgraph to an internal line --- the so-called 3-dipole move.}
\end{figure}
\begin{figure}
\centering
\def\svgwidth{0.98\columnwidth}
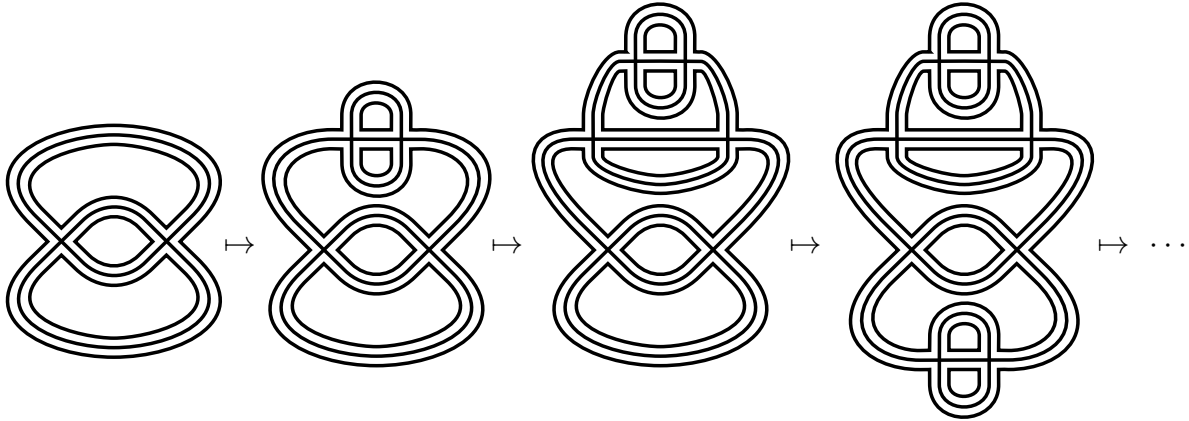
\caption{\label{fig:melons}An example of the generation of LO graphs from the elementary melon on the left via insertions of melonic 2-point subgraphs.}
\end{figure}

The leading order free energy is obtained by summing over the amplitudes of all melonic vacuum graphs. We are interested in the asymptotic behavior of this LO series around the critical value of the coupling constant $\lambda_{c,LO}$, i.e., in the limit of approaching the radius of convergence of the LO series. Let us 
state here the main result of this analysis \cite{BGRR,KOR}, which we will use in following. The LO ($\omega=0$) 2-point function $G_{LO}$ satisfies the scaling
\ba
	G_{LO}(\lambda) \sim \mt{const.} + \left( 1 - \frac{\lambda^2}{\lambda_{c,LO}^2}\right)^{\frac{1}{2}}
\ea
around the critical value $\lambda_{c,LO}$ of the coupling constant in the leading order. This is related to the behavior of the LO free energy, for which we obtain
\ba
	E_{LO}(\lambda) \sim \left( 1 - \frac{\lambda^2}{\lambda_{c,LO}^2}\right)^{2-\gamma_{LO}},
\ea
with the critical exponent, the so-called susceptibility exponent, being $\gamma_{LO}=\frac{1}{2}$.

\section{The class of next-to-leading tensor graphs}

In this section we identify, from a combinatorial and topological 
point of view, the class of NLO tensor graphs of the MO model.

\medskip
Since the MO model allows for non-orientable jackets, we have $g_J\in \mb{N}/2$, and the next-to-leading order (NLO) graphs satisfy $\omega=1/2$. A particularly simple example of such a graph is the double-tadpole $\Gamma_{2tp}$, see Fig.\ \ref{fig:2tp}, which has one non-orientable jacket with $g_J=\frac{1}{2}$. A whole class of NLO graphs can be derived from the double-tadpole via insertions of melonic two-point-functions into the propagators of $\Gamma_{2tp}$, also depicted. 
Following \cite{KOR}, let us give the following definition:

\begin{definition}\label{def:NLOcore}
	A {\bf NLO core graph} is a graph with $\omega=\frac{1}{2}$ and no melonic 2-point subgraphs.
\end{definition}
Accordingly, the double-tadpole is an NLO core graph, since it does not contain melonic 2-point subgraphs. All NLO graphs can be obtained by melonic insertions into the core graphs. Thus, the core graphs classify the NLO graphs into families related through insertions of melonic 2-point subgraphs.

To see this more clearly, let us prove the following:
\begin{lemma}\label{lem:2pf}
	Let $\Gamma$ be an MO vacuum Feynman graph, and $\Gamma'_2$ an MO two-point function subgraph of $\Gamma$. Let us denote by $\Gamma/\Gamma_2'$ the graph obtained by replacing $\Gamma'_2$ inside $\Gamma$ with a single propagator. 
We then have the relation
	\ba
		\omega(\Gamma) = \omega(\Gamma/\Gamma_2') + \omega(\Gamma') \,,
	\ea
	where $\Gamma'$ denotes the vacuum graph obtained by gluing the external legs of $\Gamma'_2$ to each other.
\end{lemma}
\begin{proof}
	Let $V'$ and $F'$ be the numbers of vertices and faces of $\Gamma'_2$, respectively. By gluing the external legs of $\Gamma'_2$ to each other we create 3 additional faces, because we loop the three strands in the legs, and thus
	\ba\label{eq:omegap}
		\omega(\Gamma')=3 + \frac{3}{2}V' - (F'+3)= \frac{3}{2}V' - F' \,.
	\ea
	On the other hand, it is clear that by contracting $\Gamma'_2$ into a propagator of $\Gamma$, we decrease the number of vertices by $V'$ and the number of faces by $F'$. Thus,
	\ba
		\omega(\Gamma/\Gamma_2') = 3 + \frac{3}{2}(V-V') - (F-F') = \omega(\Gamma) - \omega(\Gamma') 
	\ea
	by the above expression (\ref{eq:omegap}) for $\omega(\Gamma')$.
\end{proof}

\begin{figure}
\centering
\def\svgwidth{0.95\columnwidth}
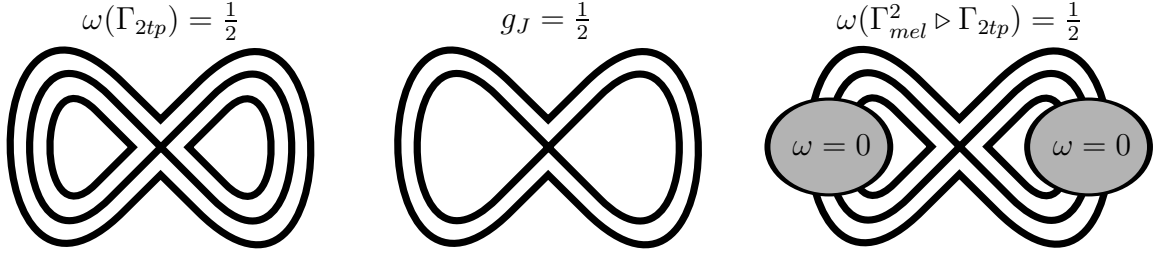
\caption{\label{fig:2tp}Left: The double-tadpole $\Gamma_{2tp}$. Middle: Its non-orientable jacket with $g_J=\frac{1}{2}$. Right: NLO graphs obtained from $\Gamma_{2tp}$ via insertions of melonic two-point functions $\Gamma^2_{mel}$.}
\end{figure}

One then has:

\begin{corollary}
	An insertion of a melonic two-point function into a propagator of an MO graph does not change the degree of the graph. In particular, all MO graphs of the form in Fig.\ \ref{fig:2tp}, right-hand-side, have $\omega=1/2$, and are therefore NLO.
\end{corollary}

Let us now give the following definition:

\begin{definition}
A graph $\Gamma$ is {\bf 2-particle-irreducible (2PI)}, if it cannot be disconnected by cutting any two lines, or equivalently, if $\Gamma$ does not contain any proper non-trivial (other than the propagator) two-point function subgraphs.
\end{definition}

One then has:

\begin{lemma}
	A NLO core graph of the MO model is 2-particle-irreducible.
\end{lemma}
\begin{proof}
By definition, a core graph does not have melonic subgraphs. Note that the only non-trivial two-point functions with $\omega=0$ are the melonic ones. Now, let $\Gamma'_2$ be a proper non-trivial two-point function subgraph of $\Gamma$. It then follows from Lemma \ref{lem:2pf} that one cannot have a proper non-trivial two-point function subgraph $\Gamma'_2$ in an NLO core graph $\Gamma$:
	\begin{itemize}
	\item[(i)] If $\omega(\Gamma')=0$. $\Rightarrow$ $\Gamma$ is not a core graph.
	\item[(ii)] If $\omega(\Gamma')\geq \frac{1}{2}$, either 
		\begin{itemize}
		\item $\omega(\Gamma/\Gamma_2')=0$. $\Rightarrow$ $\Gamma$ is not a core graph, or otherwise
		\item $\omega(\Gamma/\Gamma_2')\geq\frac{1}{2}$, so we have $\omega(\Gamma) = \omega(\Gamma/\Gamma'_2) + \omega(\Gamma') \geq 1$ $\Rightarrow$ $\Gamma$ is not NLO.
		\end{itemize}
	\end{itemize}
	Thus, none of the possibilities for a proper non-trivial two-point function subgraph $\Gamma'_2$ allow $\Gamma$ to be an NLO core graph and 
this concludes the proof.
\end{proof}

What remains to be proven, then, is that the double-tadpole is the \emph{only} NLO core graph:

\begin{proposition}\label{prop:NLOcore}
	The only NLO core graph of the MO model is the double-tadpole of Fig. \ref{fig:2tp}.
\end{proposition}
\begin{proof}
	In order to have $\omega=1/2$, one of the jackets must have genus $\frac{1}{2}$, and the two others must have genus $0$. In particular, the jacket formed by the outer strands is always orientable \cite{Stephane}, so it must be planar, genus $0$. Thus, the jacket with genus $\frac{1}{2}$ is one of the jackets containing the inner stands.
	
Now, from the planarity of the outer jacket, we get
\bas
	0=1 - \frac{1}{2}(F_o -V) \quad \Rightarrow \quad F_o = V + 2 \,,
\eas
where $F_o$ is the number of faces formed by the outer strands. Substituting this into (\ref{eq:omega}), we get
\ba\label{eq:inneromega}
	\omega = 3 + \frac{3}{2}V - (F_o + F_i) = 1 + \frac{1}{2}V - F_i \,,
\ea
where $F_i$ is the number of faces formed by the inner strands.

Now, let us focus on the faces formed by the inner strands, {\it inner faces}. Notice that the planarity of the outer jacket implies that the graph can be drawn on a plane, so that the inner faces only cross at vertices. Let us assume, for simplicity, the graph to be drawn in this way, so that now the vertices of the graph and the crossings of inner faces are in one-to-one correspondence.

In fact, each inner face of a NLO core graph with $F_i\geq 2$ must have at least four crossings in total with the other inner faces. Note that closed faces on a plane always cross an even number of times. Thus, to prove this claim, we provide two simple proofs by contradiction for the cases of zero and two crossings.

Case 0: Assume that an inner face $f$ in a NLO core graph with $F_i\geq 2$ does not cross the other inner faces at all. Clearly, this makes the graph disconnected, as the face $f$ corresponds to a disjoint connected component of the graph. Therefore, the graph cannot be a NLO core graph, and we have reached a contradiction.

Case 2: Assume that an inner face $f$ in a NLO core graph with $F_i\geq 2$ crosses the other inner faces twice in total. Clearly, the face $f$ must cross the same inner face twice, since closed faces on a plane always cross an even number of times. We may further assume that $f$ does not cross itself, since we are eventually interested in obtaining a lower limit on the number of vertices, while self-crossings of $f$ increase the number of vertices. Accordingly, the face $f$ divides the plane into two separate regions with only two inner strands crossing $f$, corresponding to two vertices along $f$. Thus, the graph can be cut into two two-point subgraphs by cutting, just inside (or outside) of the inner face $f$, the two lines whose inner strands $f$ crosses at these vertices (see Fig.\ \ref{fig:2crossing} for an illustration).
\begin{figure}
\centering
\def\svgwidth{0.4\columnwidth}
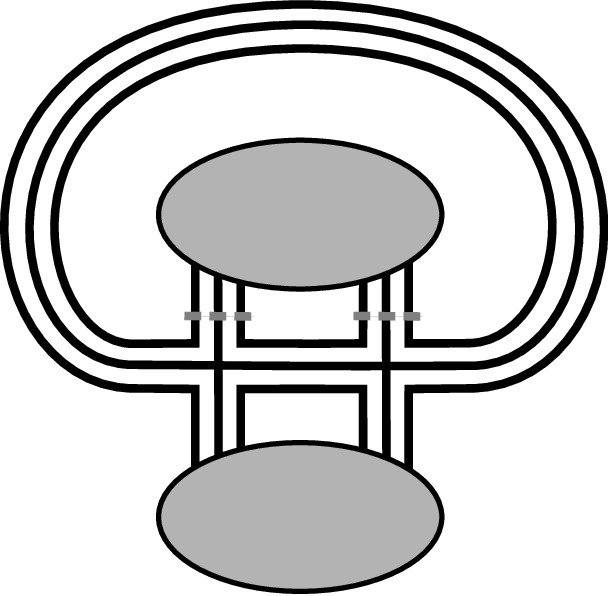
\caption{\label{fig:2crossing}If an inner face crosses only two strands, the graph is decomposable into two two-point functions by cutting along the grey dashed lines.}
\end{figure}
The resulting two-point function inside (or outside) $f$ can either be
\begin{itemize}
\item	a bare propagator $\Rightarrow$ the inner face $f$ is a part of a 
melonic two-point function, or 
\item a non-trivial two-point function $\Rightarrow$ the graph is not 2PI.
\end{itemize}
Thus, the graph cannot be a NLO core graph, and we have reached a contradiction.

The above implies that for a NLO core graph with $F_i\geq 2$ each inner face must cross the other inner faces at least four times in total. Since each crossing is shared by two faces, the number of vertices must satisfy $V\geq 2F_i$ for an NLO core graph with $F_i\geq 2$. Thus, we get from (\ref{eq:inneromega}) above that $\omega\geq 1$, when $F_i\geq 2$, and so we must have $F_i=1$. Moreover, again using (\ref{eq:inneromega}), $F_i=1$ implies $V=1$ for $\omega=\frac{1}{2}$. We thus conclude that the double-tadpole is the only vacuum MO graph with $V=1$.

\end{proof}

\medskip

Let us now state the main result of this section:

\begin{theorem}\label{thm:main}
	All NLO graphs of the MO tensor model are of the form Fig.\ \ref{fig:2tp}, right-hand-side.
\end{theorem}
{\begin{proof} 
The proof follows directly from Proposition \ref{prop:NLOcore} above and from Definition \ref{def:NLOcore} of an NLO core graph.  
 \end{proof}

To conclude the section, let us recall that in the case of colored tensor models any 1-dipole or $d$-dipole move acting on an NLO graph leads to another NLO graph \cite{Stephane,GS}; a similar situation appears also for the so-called ``uncolored'' tensor model \cite{DGR}. 
This is different for the MO\ case we treat here. 
In particular, it is not clear how exactly one could extend the notion of a 1-dipole move to MO tensor graphs in the absence of color labels. One could still trivially consider a weak analogue of the 1-dipole move for the class of MO graphs as illustrated in Fig.\ \ref{fig:1dip}, but this does not conserve the degree of a graph.
\begin{figure}
\centering
\def\svgwidth{0.75\columnwidth}
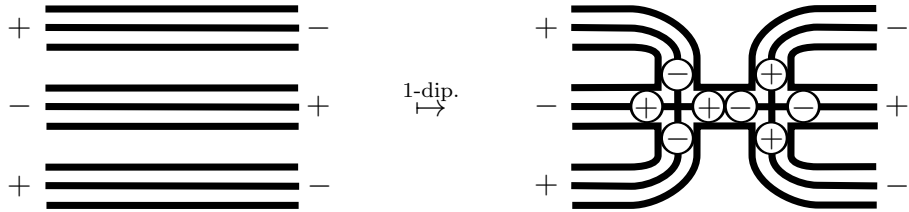
\caption{\label{fig:1dip}The analogue of a 1-dipole move of the colored tensor model. Here we have kept track of the orientations of the edges by inserting $\pm$'s on the edges emanating from vertices.}
\end{figure}
For the sake of completeness, let us give in Fig.\ \ref{fig:1dipex} an example of this analogue 1-dipole move {\it inequivalent to creation/annihilation of elementary melons} acting on an NLO MO graph, which does not lead to an NLO MO graph. It may be expected that the lack of colors and thus the notion of dipole moves will complicate the classification of MO graphs in higher orders of the large $N$ expansion as compared to colored models.
\begin{figure}
\centering
\def\svgwidth{0.8\columnwidth}
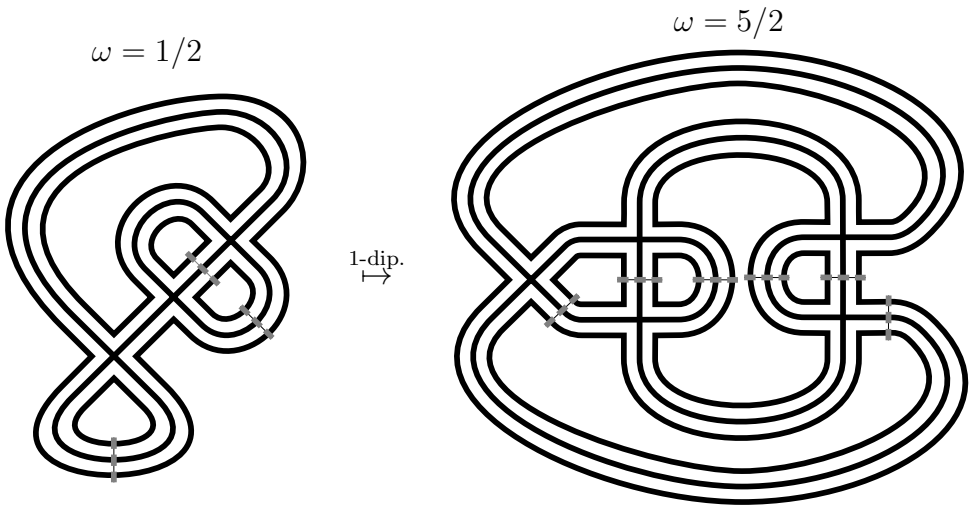
\caption{\label{fig:1dipex}An example of an analogue `1-dipole move' applied to an NLO graph (to the lines marked with grey dashed line) resulting in a non-NLO graph ($\omega=\frac{5}{2}$).}
\end{figure}

We should also emphasize that the NLO sector of the colored tensor model, considered in \cite{KOR}, corresponds to the case $\omega=1$ (for $d=3$), whereas we have here considered the case $\omega=\frac{1}{2}$, which constitutes the NLO sector for the MO tensor model. Thus, the two NLO sectors are not directly comparable. However, using Lemma \ref{lem:2pf}, one can easily see that there are non-colorable graphs included in the $\omega=1$ sector of the MO model, as one obtains a core graph with $\omega=1$ by inserting a tadpole 2-point subgraph into the double-tadpole. We expect there to appear many more such non-colorable core graphs in higher integer orders, thus modifying the results of \cite{KOR}. However, for the remainder of this paper, we will concentrate only on the $\omega=\frac{1}{2}$ NLO sector of the MO model, while work on the higher order sectors is under way.

\section{NLO series: radius of convergence \& susceptibility exponent}

In order to study the behavior of the NLO series, following \cite{KOR}, we will 
 study the connected NLO two-point function. The graphs contributing to the connected NLO two-point function can be obtained from the NLO vacuum graphs by cutting any one of the internal lines of an NLO vacuum graph. Thus, we can in a straightforward
manner import the classification of NLO vacuum graphs obtained in the previous section to the case of connected NLO two-point graphs. We will express the NLO 
two-point function in terms of the LO two-point function through algebraic identities relating the LO and NLO two-point functions. More specifically, any 
two-point function is of the form bare propagator multiplied by a specific function. We will denote this function associated to the connected LO two-point function as $G_{LO}$, the function associated to the connected NLO two-point function as $G_{NLO}$, and the function associated to the 1-particle-irreducible (1PI) NLO two-point function as $\Sigma_{NLO}$. The identities between these different functions arise diagrammatically as the classical QFT identities of Feynman graphs.

The first identity for the two-point functions, illustrated in Fig.\ \ref{fig:gnlo}, states that the connected NLO two-point function $G_{NLO}$ is obtained by gluing connected LO two-point functions $G_{LO}$ on both sides of the 1PI NLO two-point function $\Sigma_{NLO}$.
\begin{figure}
\centering
\def\svgwidth{0.75\columnwidth}
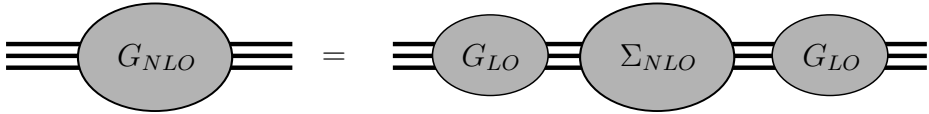
\caption{\label{fig:gnlo}The connected NLO two-point function is obtained by gluing connected LO two-point functions on both sides of the 1PI NLO two-point function.}
\end{figure}
This may be expressed as:
\ba
	G_{NLO} = G_{LO}^2 \Sigma_{NLO} \,.
\ea
The second identity is obtained from two different ways of obtaining 1PI NLO two-point graphs from connected LO and NLO two-point graphs, as in Fig.\ \ref{fig:Snlo}, and writes
\ba
	\Sigma_{NLO} = \lambda G_{LO} + 3\lambda^2 G_{LO}^2 G_{NLO} \,,
\ea
where the combinatorial factor of 3 arises from the three different internal lines of the elementary melon, on which the NLO two-point function can be inserted. These two possibilities exhaust all contributions to the 1PI NLO two-point function, since any such graph contains only one tadpole due to Theorem \ref{thm:main}, and thus either all melonic subgraphs are inside the tadpole (the first summand in Fig.\ \ref{fig:Snlo}) or the tadpole is a subgraph of a melonic graph (the second summand in Fig.\ \ref{fig:Snlo}).
\begin{figure}
\centering
\def\svgwidth{0.9\columnwidth}
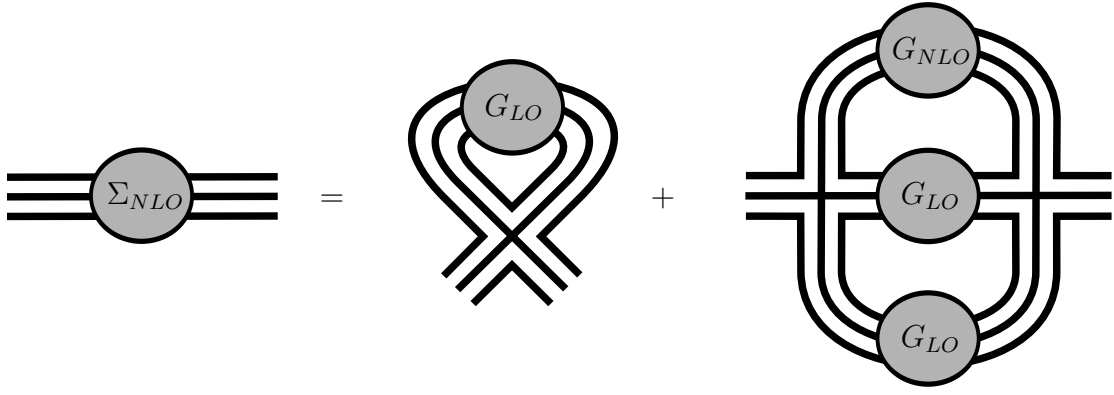
\caption{\label{fig:Snlo}The 1PI NLO two-point function may be obtained in two ways from connected LO and NLO two-point functions.}
\end{figure}

Now, putting these two identities together, we obtain first the identity
\ba
	G_{LO}^{-2} G_{NLO} = \lambda G_{LO} + 3\lambda^2 G_{LO}^2 G_{NLO} \,,
\ea
from which we can solve for $G_{NLO}$ in terms of the connected LO two-point function
\ba
	G_{NLO} = \frac{\lambda G_{LO}^3}{1-3\lambda^2 G_{LO}^4} \,.
\ea
On the other hand, differentiating the LO two-point function relation $G_{LO} = 1 + \lambda^2 G_{LO}^4$ we get
\ba
	\frac{\prt}{\prt\lambda} G_{LO} = \frac{2\lambda G_{LO}^4}{1 - 4\lambda^2 G_{LO}^3} = \frac{2\lambda G_{LO}^5}{1 - 3\lambda^2 G_{LO}^4} \,,
\ea
where for the last equality we used the LO two-point function identity again. Thus, we get the expression
\ba
	G_{NLO} = \frac{\lambda}{G_{LO}^2} \frac{\prt}{\prt\lambda^2} G_{LO} \,,
\ea
which implies, together with $G_{LO} \sim \mt{const.}\ + (1 - (\lambda^2/\lambda_c^2))^{1/2}$,
\ba
	G_{NLO} \sim \left(1 - \frac{\lambda^2}{\lambda_c^2}\right)^{-1/2} \,.
\ea

Finally, we obtain from the Schwinger-Dyson equation \cite{KOR}
\ba
	0 = \int \dd\bar{\phi}\ \dd\phi\ \frac{\delta}{\delta\phi_{ijk}} \left( \phi_{i'j'k'} e^{-S[\phi,\hat{\phi}]} \right)
\ea
the relation
\ba\label{eq:SDeq}
	G_{NLO} = 1 - 4\lambda^2 \frac{\prt}{\prt\lambda^2} E_{NLO}
\ea
relating the connected two-point function $G_{NLO}$ to the free energy $E_{NLO}$. Accordingly, we have $E_{NLO} \sim (1 - (\lambda/\lambda_c)^2)^{1/2}$ from (\ref{eq:SDeq}), and thus find the same critical value of the coupling constant (i.e., the radius of convergence) for the NLO series as for the LO series.
Nevertheless, one has a distinct value for the NLO susceptibility exponent (or the critical exponent) 
\ba\label{main}
\gamma_{NLO}=\frac 32.
\ea
These properties indicate the possibility for the existence of a double-scaling limit for the MO\ model, as we will discuss in the following section.

\section{Perspectives --- towards a double-scaling limit}

This paper analyses in detail the next-to-leading order in the large $N$ expansion of 
multi-orientable random tensor models. It thus represents the first step towards the implementation of 
a double-scaling limit for the multi-orientable model.
It appears to us that the most important perspective for our study is thus 
the implementation of such a double-scaling limit appropriate for the model.

Let us recall here that the celebrated double-scaling limit of matrix models consists 
of taking, in a correlated way, the double limit $N\to\infty$ and $\lambda\to\lambda_c$ (where 
$\lambda_c$ is some critical value of the coupling constant).
This allows for non-vanishing genus topologies to count.  
Switching to a gravitational interpretation \`a la random dynamical triangulations, where the two constants of the model are 
related to Newton constant and to the cosmological constant, 
the large-$N$ limit corresponds to the vanishing limit of Newton constant, 
while the limit $\lambda\to\lambda_c$ corresponds to the large volume limit. 
Thus, for matrix models, the double scaling limit mechanism allows to access the regime of finite 
(non-vanishing) Newton constant.
(We refer the interested reader to 
\cite{double-matrix} or, for a general review on matrix models, to \cite{DGZ}.) For tensor models, it is clear that some sort of multiple-scaling limit is necessary to access this regime, but the double-scaling limit only picks up specific contributions at each order of the 1/N expansion, and its physical interpretation has not yet been fully clarified.

As already mentioned in the introduction, the double scaling limit of the colored \cite{gurau} and respectively ``uncolored'' \cite{uncolored} tensor models 
have been implemented in \cite{GS} and respectively \cite{DGR}.
It is worth noticing that in \cite{KOR}, the critical susceptibility exponent for the NLO of the large $N$ expansion 
was found to be $3/2$, which is the same result as the one we obtained in the previous 
section (see equation \eqref{main} above). Nevertheless, 
in \cite{DGR} it was found that the non-perturbative result (using 
the so-called loop-vertex expansion \cite{LVE} method) for this critical susceptibility exponent is again $1/2$, 
the value obtained for the leading melonic order. It was then argued by the authors of \cite{DGR}
that this does not represent any discrepancy with the previous results of \cite{KOR},
since double-scaling resums an infinity of contributions, and the one found in \cite{KOR}
is just the second order term in the infinite sum.

It thus appears particularly interesting to us to investigate whether or not such a phenomenon appears 
also in the case of the multi-orientable models and if the non-perturbative value of the critical 
susceptibility exponent is again $1/2$ or not. The result will dramatically depend on the role played by the sectors with half-integer degree.

\section*{Acknowledgements}
The authors acknowledge Vincent Rivasseau and Razvan Gurau for
discussions, and would also like to thank Stephane Dartois for a careful reading of the manuscript.
Moreover, the authors 
acknowledge the ``Combinatoire  alg\'ebrique'' Univ.\ Paris 13, Sorbonne Paris Cit\'e BQR  grant and the ``Cartes 3D'' CNRS PEPS grant.
A.\ Tanasa further acknowledges
the grants PN 09 37 01 02 and CNCSIS Tinere Echipe 77/04.08.2010. M.\ Raasakka's research is supported by Emil Aaltonen Foundation.

\noindent
{\it\small $(a)$ LIPN, Institut Galil\'ee, CNRS UMR 7030, 
Universit\'e Paris 13, Sorbonne Paris Cit\'e,}\\{\it\small 99 av. Clement, 93430 Villetaneuse, France, EU}\\
{\it\small $(b)$ Horia Hulubei National Institute for Physics and Nuclear Engineering,
P.O.B. MG-6, 077125 Magurele, Romania, EU}
\end{document}